\documentclass[12pt]{amsart}
\usepackage{amssymb, amsmath}
\usepackage{fourier}

\theoremstyle{plain}

\newtheorem{theorem}{Theorem}
\newtheorem*{Definition}{Definition}

\newtheorem{lemma}[theorem]{Lemma}
\newtheorem{claim}[theorem]{Claim}
\newtheorem{corollary}[theorem]{Corollary}

\theoremstyle{remark}
\newtheorem*{Remark}{Remark}
\newtheorem*{Example}{Example}

\newcommand{\be}{\begin{equation}}
\newcommand{\ee}{\end{equation}}
\newcommand{\ba}{\begin{array}}
\newcommand{\ea}{\end{array}}

\newcommand{\mc}[1]{\mathcal{#1}}
\newcommand{\supp}[1]{\mathrm{supp}\left(#1\right)}
\newcommand{\rnk}[1]{\mathrm{rank}\left(#1\right)}
\newcommand{\spn}[1]{\mathrm{span}\left(\left\{ #1 \right\}\right)}
\newcommand{\abs}[1]{\left|#1\right|}

\newcommand{\eps}{\varepsilon}

\newcommand\cB{\mathcal B}
\newcommand\cH{\mathcal H}
\newcommand\cP{\mathcal P}
\newcommand\cS{\mathcal S}
\newcommand\C{\mathbb C}
\newcommand\R{\mathbb R}

\newcommand\tr{\operatorname{Tr}}
\newcommand{\ler}[1]{\left( #1 \right)}

\newcommand\B{\mathcal B}

\begin{document}

\title[]{Maps on quantum states preserving Bregman and Jensen divergences}

\author{D\'aniel Virosztek}
\address{Department of Analysis, Institute of Mathematics\\
Budapest University of Technology and Economics\\
H-1521 Budapest, Hungary
and
MTA-DE ``Lend\" ulet'' Functional Analysis Research Group, Institute of Mathematics\\
         University of Debrecen\\
         H-4002 Debrecen, P.O. Box 400, Hungary}
\email{virosz@math.bme.hu}
\urladdr{http://www.math.bme.hu/\~{}virosz}

\thanks{
The author was supported by the ``Lend\" ulet'' Program (LP2012-46/2012) of the Hungarian Academy of Sciences and by the Hungarian Scientific Research Fund (OTKA) Reg. No.  K104206. The author was also supported by the ``For the Young Talents of the Nation'' scholarship program (NTP-EF\"O-P-15-0481) of the Hungarian State.}

\dedicatory{Dedicated to Anna Gelniczky}


\begin{abstract}
We describe the structure of the bijective transformations on the set of density operators which preserve the Bregman $f$-divergence for an arbitrary differentiable strictly convex function $f.$ Furthermore, we determine the preservers of the Jensen $f$-divergence in the case when the generating function $f$ belongs to a recently introduced function class called Matrix Entropy Class.
\end{abstract}
\maketitle
\section{Introduction}
\subsection{Motivations and overview of the literature}
The investigation of certain measures of dissimilarity between two objects (numbers, vectors, matrices, functions and so on) plays an essential role in several areas of mathematics ans mathematical physics. Some of the widely investigated measures are distance functions, but there are many important measures which do not satisfy the properties of distance.
\par

For instance, the \emph{square loss function} has been used widely for regression analysis, \emph{Kullback-Leibler divergence} \cite{kl} has been applied to compare two probability density functions, the \emph{Itakura-Saito divergence} \cite{is} is used as a measure of the perceptual difference between spectra. The \emph{Bregman divergence} was introduced by Lev Bregman \cite{breg} for convex functions $\phi :\R^d\to\R$ as the $\phi$-depending nonnegative measure of discrepancy of elements of $\R^d$. Originally his motivation was the problem of convex programming, but it became widely researched both from theoretical and practical viewpoints. The remarkable fact that all the aforementioned divergences --- let alone \emph{Stein's loss} and \emph{Umegaki's relative entropy} --- are special cases of the Bregman divergence shows its importance \cite{ba05}.

\par

Moreover, in the recent volume \cite{nb} on matrix information geometry 3 chapters are devoted to the study of Bregman divergences. One feature of \emph{Jensen divergences} which justifies their importance is that Bregman divergences can be considered as asymptotic Jensen divergences (see Section 6.2 in \cite{nb}). Furthermore, the celebrated \emph{Jensen-Shannon divergence} and its non-commutative counterpart the \emph{Jensen-von Neumann divergence} are particular Jensen divergences.
\par
For some recent results conserning Bregman divergences of operators we refer to \cite{ls14} and \cite{pv15}. In \cite{ls14}, M. Lewin and J. Sabin characterized a certain monotonicity property of the Bregman divergence by the operator monotonicity of the derivative of the corresponding scalar function. In \cite{pv15}, J. Pitrik and the author gave a characterization of the joint convexity of the Bregman divergence in terms of the generating function.
\par
The Bregman and Jensen  divergences are generalized distance measures.
This latter notion stands for any function $d:\mathcal X\times \mathcal X\to [0,\infty )$ on any set $\mathcal X$ with the mere property that for $x,y\in \mathcal X$ we have $d(x,y)=0$ if and only if $x=y$. Transformations which preserve generalized distance measures are called generalized isometries.
\par

In a recent paper L. Moln\'ar, J. Pitrik and the author determined the structure of the generalized isometries of the cone of positive definite matrices with respect to Bregman divergences and Jensen divergences \cite{mpv15}.
Let us mention that in the papers \cite{ml13j,ml15b} L. Moln\'ar and his coauthors considered a certain family of generalized distance measures on the cone of positive definite matrices and the structure of the isometries was obtained by Mazur-Ulam type arguments. The aforementioned family of divergences is almost disjoint from the families of Bregman divergences and Jensen divergences. The intersection contains only the Stein's loss and the Chebbi-Moakher log-determinant $\alpha$-divergences, respectively.

\subsection{Goals}

In this paper we describe the structure of the generalized isometries of the set of density matrices --- which used to represent the state space of a finite quantum system --- with respect to Bregman and Jensen divergences. It turns out that every bijective transformation which leaves the Bregman or Jensen divergence invariant is implemented by a unitary or antiunitary operator on the underlying Hilbert space. Such a result may be considered as a Wigner type result. For other closely related Wigner type results we refer to \cite{lm08} and \cite{mnsz13}. In fact, at several points of our argument we use ideas and techniques of the latter two papers.

\subsection{Basic notions and notations}

Throughout this paper the following notations will be used. $\mc{H}$ stands for a finite dimensional complex Hilbert space. $\R^+$ ($\R^{++}$) consists of all nonnegative (positive) numbers and $\mc{B}(\mc{H})$ ($\mc{B}^{sa}(\mc{H}),$ $\mc{B}^{+}(\mc{H}),$ $\mc{B}^{++}(\mc{H})$) denotes the set of bounded (self-adjoint, positive semidefinite, positive definite) linear operators on the Hilbert space $\mc{H}.$
$\cS\ler{\cH}$ stands for the state space of $\cH$ (the set of positive semidefinite operators with unit trace) and $\cP_1(\cH)$ denotes the set of rank-one projections on $\cH.$


If $f: \, I \rightarrow \R$ is a function defined on an interval $I \subset \R$ then the corresponding \emph{standard operator function} is the following map:
$$
f: \{A \in \mc{B}^{sa}(\mc{H}): \ \sigma(A) \subseteq I \} \rightarrow \mc{B}(\mc{H})
$$
$$
A=\sum_{a \in \sigma(A)} a P_a \mapsto f(A):=\sum_{a \in \sigma(A)} f(a) P_a,
$$
where $\sigma(A)$ is the spectrum and $P_a$ is the spectral projection corresponding to the eigenvalue $a.$

\subsection{Bregman divergences on positive definite operators}
Let $f$ be a differentiable strictly convex function on $(0,\infty)$. (Note that the derivative of $f$ is necessarily contiuous.) The Bregman $f$-divergence of the positive definite operators $A, B \in \B^{++}(\cH)$ is defined by
$$
H_f(A,B)= \tr \ler{f(A)-f(B)-f'(B)(A-B)},
$$
see e. g. formula $(5)$ in \cite{pv15}. Easy computation shows that if the spectral decompositions are
$$
A=\sum_{a \in \sigma(A)} a P_a \text{ and } B=\sum_{b \in \sigma(B)} b Q_b.
$$
then we have
$$
H_f(A,B)=\sum_{a \in \sigma(A), b \in \sigma(B)} \ler{f(a)-f(b)-f'(b)(a-b)} \tr P_a Q_b.
$$

\subsubsection{The extension to positive semidefinite operators}
If $f$ can be extended to $0$ by continuity, then the Bregman $f$-divergence can be extended to positive semidefinite operators by continuity the following way. Let
\be \label{limesz}
H_f\ler{X, Y}:= \lim_{\eps \to 0} \ler{X+\eps I, Y + \eps I}
\ee
for positive semidefinite operators $X$ and $Y.$ In the followings we show that the limit (\ref{limesz}) always exists and takes values in $\R^+ \cup \{+\infty\}.$ The argument is based on \cite[Sec. 2.1]{pv15}.
\par
If $X$ and $Y$ admit the spectral decompositions $X=\sum_{x \in \sigma(X)} x P_x$ and $Y=\sum_{y \in \sigma(Y)} y Q_y$ then
$$
H_f(X+\eps I,Y+\eps I)
$$
\be \label{kiir}
=\sum_{x \in \sigma(X), y \in \sigma(Y)} \ler{f(x+\eps)-f(y+\eps) -f'(y+\eps)(x-y)} \tr P_x Q_y.
\ee
Assume that $f \in C^1((0,\infty)) \cap C^0([0,\infty)),$ that is, $\lim_{x \to 0} f(x) \in \R.$
The convexity of $f$ gives that $f'$ is monotone increasing, hence $\lim_{\eps \to 0} f'(\eps) \in \R$ or $\lim_{\eps \to 0} f'(\eps)=-\infty.$
\par

Clearly, if $\lim_{\eps \to 0} f'(\eps) \in \R,$ then the limit of (\ref{kiir}) is a real number. If $\supp{X} \subseteq \supp{Y},$ then for every $x \in \sigma(X)$ we have $\tr P_x Q_0 = 0$ or $x=0,$ hence the limit of (\ref{kiir}) is finite in this case, as well. It is easy to see that if $\supp{X} \nsubseteq \supp{Y}$ and $\lim_{\eps \to 0} f'(\eps)=-\infty,$ then the limit of (\ref{kiir}) is $+\infty.$

\subsubsection{Computation rules}
By the above argument, if $\lim_{x \to 0} f'(x)= - \infty,$ for positive semidefinite operators $X$ and $Y$ the following computation rule holds.
\be \label{comp}
H_f(X,Y)=\sum_{x \in \sigma(X), y \in \sigma(Y)\setminus\{0\}} \ler{f(x)-f(y) -f'(y)(x-y)} \tr P_x Q_y 
\ee
\be \label{compvar}
=\tr_{\supp{Y}} \ler{f(X)-f(Y)-f'(Y)(X-Y)}
\ee
if $\supp{X} \subseteq \supp{Y}$
and
$$
H_f(X,Y)=\infty
$$
if $\supp{X} \nsubseteq \supp{Y}.$
(For any $\mathcal{K} \subset \cH,$ $\tr_\mathcal{K}$ means that we take the trace only on the subspace $\mathcal{K}.$)

\par

If $\lim_{x \to 0} f'(x) \in \R$ then the computation rule is simply the following.
\be \label{comp_2}
H_f(X,Y)=\sum_{x \in \sigma(X), y \in \sigma(Y)} \ler{f(x)-f(y) -f'(y)(x-y)} \tr P_x Q_y
\ee
\be \label{comp_2var}
=\tr \ler{f(X)-f(Y)-f'(Y)(X-Y)}.
\ee

\begin{Example}
For the standard entropy function $f(x)=x \log{x}$ the induced Bregman $f$-divergence on density matrices is the \emph{Umegaki relative entopy}
$$
H_f(A,B)=\tr A \ler{\log{A}-\log{B}}
$$
which is one of the most important numerical quantities in quantum information theory. Therefore, Bregman $f$-divergences may be considered as genralized relative entropies \cite{ls14}.
\par
For any $q>1,$ the function $f_q: x\mapsto f_q(x):=\frac{x^q-x}{q-1}$ is convex, and the induced Bregman divergence is
$$
H_{f_q}\ler{A,B}=\tr B^q+\frac{1}{q-1}\ler{\tr A^q-q\tr AB^{q-1}},
$$
see \cite{pv15}.
In the particular case $q=2$ the latter quantity is just the square of the Hilbert-Schmidt norm,
$$
H_{f_2}\ler{A,B}=\tr \ler{A-B}^2.
$$
\end{Example}

\subsection{Jensen divergences on positive semidefinite operators} \label{jendef}
Let $f$ be a strictly convex function on $(0,\infty)$ such that the limit $\lim_{x \to 0+} f(x)=:f(0)$ exists. The Jensen $f$-divergence of the positive semidefinite operators $A \in \cB^{+}(\cH)$ and $B \in \cB^{+}(\cH)$ is defined by
$$
J_{f}(A,B)=\tr \ler{\frac{1}{2} \ler{f(A)+f(B)}-f\ler{\frac{1}{2} \ler{A + B}}},
$$
see e. g. \cite{mpv15}.
We investigate the preservers of the Jensen $f$-divergence in the case when the generating function $f$ belongs to the \emph{Matrix Entropy Class}.

In the recent paper \cite{chtr} \emph{Tropp} and \emph{Chen} defined the Matrix Entropy Class the following way.

\begin{Definition} 
The Matrix Entropy Class consists of the real valued functions defined on $[0,\infty)$ that are either affine or satisfy the following conditions.
\begin{itemize}
 \item $f$ is convex and $f \in C([0, \infty)) \cap C^2((0, \infty)).$
 \item For every finite dimensional Hilbert space $\cH$ the map
 $$ \cB(\cH)^{++} \rightarrow \cB\ler{\cB(\cH)^{sa}}; \, \, X \mapsto \ler{\mathbf{D} f'[X]}^{-1}$$
 is concave with respect to the semidefinite order,
 where $\mathbf{D} f[Y]$ denotes the Fr\'echet derivative of the standard operator function $f: \cB(\cH)^{++} \rightarrow \cB(\cH)^{sa}$ at the point $Y.$
 
\end{itemize}
\end{Definition}

\begin{Example}
The standard entropy function $f(x)=x \log{x}$ is an important element of the Matrix Entropy Class \cite{chtr}. The induced Jensen $f$-divergence is the well-known \emph{Jensen-von Neumann divergence}
$$
J_f(A,B)=\frac{1}{2}\ler{\tr A \log{A}+\tr B \log{B}}-\tr\ler{\frac{A+B}{2}}\log{\ler{\frac{A+B}{2}}}.
$$

If $1<q\leq2,$ then the function $f_q: x\mapsto f_q(x):=\frac{x^q-x}{q-1}$ belongs to the Matrix Entrpy Class \cite{chtr}, and the induced Jensen divergence is
$$
J_{f_q}\ler{A,B}=\frac{1}{q-1}\ler{\frac{\tr A^q+\tr B^q}{2}-\tr \ler{\frac{A+B}{2}}^q}.
$$
In particular, if $q=2,$ then we have
$$
J_{f_2}\ler{A,B}=\tr \ler{\frac{A-B}{2}}^2.
$$
\end{Example}
\section{The main results}
It is clear that any unitary or antiunitary conjugation leaves the Bregman divergences and Jensen divergences invariant.
The main result of this paper is that the converse statement is also true, i. e., the preservers of Bregman and Jensen divergences are necessarily unitary or antiunitary conjugations.

\begin{theorem} \label{fo_breg}
Let $f \in C^1((0, \infty)) \cap C^0([0,\infty))$ be a strictly convex function.
Let $\phi: \cS(\cH) \rightarrow \cS(\cH)$ be a bijection which preserves the Bregman $f$-divergence, that is,
$$
H_f(\phi(A),\phi(B))=H_f(A,B) \qquad \ler{A, B \in \cS(\cH)}.
$$
Then there exists a unitary or antiunitary transformation $U: \cH \rightarrow \cH$ such that
$$
\phi(A)=UAU^* \qquad \ler{A \in \cS(\cH)}.
$$
\end{theorem}

\begin{theorem} \label{fo_jen}
Let $f$ be a strictly convex element of the \emph{Matrix Entropy Class}.
Let $\phi: \cS(\cH) \rightarrow \cS(\cH)$ be a bijection which preserves the Jensen $f$-divergence, that is,
$$
J_f(\phi(A),\phi(B))=J_f(A,B) \qquad  \ler{A, B \in \cS(\cH)}.
$$
Then there exists a unitary or antiunitary transformation $U: \cH \rightarrow \cH$ such that
$$
\phi(A)=UAU^* \qquad  \ler{A \in \cS(\cH)}.
$$
\end{theorem}

\section{Proofs}

\begin{Remark}
Affine perturbation of the generating function does not change the Bregman or the Jensen divergence, that is
$$
H_{f+a}(.,.)=H_f(.,.)
$$
and
$$
J_{f+a}(.,.)=J_f(., .)
$$
for any convex function $f$ and affine function $a(x)=\alpha x + \beta$. Therefore in the followings we may and do assume that $f(0)=f(1)=0.$
\end{Remark}

\subsection{The proof of Theorem \ref{fo_breg}}

\begin{proof}[Case I]
First we investigate the case when $\lim_{x \to 0} f'(x)= - \infty.$
\par
In this first part of the proof we basicly follow the argument of \cite{lm08}, but the more general statement requires new techniques at some crucial points of the proof.
\par
As $f'(x)$ is unbounded from below, the divergence $H_f(A,B)$ is finite if and only if $\supp{A} \subseteq \supp{B}.$ Therefore, any divergence-preserving transformation $\phi$ has the following properties.
\be \label{tart1}
\supp{A} \subseteq \supp{B} \Leftrightarrow \supp{\phi(A)} \subseteq \supp{\phi(B)},
\ee
\be \label{tart2}
\supp{A} = \supp{B} \Leftrightarrow \supp{\phi(A)} = \supp{\phi(B)}
\ee
and
\be \label{tart3}
\supp{A} \subsetneq \supp{B} \Leftrightarrow \supp{\phi(A)} \subsetneq \supp{\phi(B)}.
\ee
As a consequence, $\phi$ preserves the rank --- the reader should consult \cite{lm08} for a more detailed argument. In particular, the image of a rank-one projection is a rank-one projection, as well. So $\phi$ restricted to $\cP_1(\cH)$ is a bijection from $\cP_1(\cH)$ to $\cP_1(\cH)$.
\par
Let $P$ and $Q$ be orthogonal elements of $\cP_1(\cH),$ set $0<\lambda<\mu<1$ such that $\lambda+\mu=1,$ $S:=\lambda P + \mu Q$ and let $R \in \cP_1(\cH).$
\par
If $\supp{R}\subseteq \supp{S}$ then by the computation rule (\ref{comp}) we have
$$
H_f(R,S)
$$
$$
=\ler{f(1)-f(\lambda)-f'(\lambda)(1-\lambda)}\tr R P+ \ler{f(0)-f(\lambda)-f'(\lambda)(0-\lambda)}\tr(I-R) P
$$
$$
+\ler{f(1)-f(\mu)-f'(\mu)(1-\mu)}\tr R Q+ \ler{f(0)-f(\mu)-f'(\mu)(0-\mu)}\tr(I-R)Q
$$
$$
=-f'(\lambda) \tr RP -f'(\mu) \tr RQ +\lambda f'(\lambda)-f(\lambda)+\mu f'(\mu)-f(\mu)
$$
\be \label{szamol}
=-f'(\lambda) \tr RP -f'(\mu) \tr RQ +C_{f, \lambda, \mu}
\ee
if we introduce the notation $C_{f, \lambda, \mu}=\lambda f'(\lambda)-f(\lambda)+\mu f'(\mu)-f(\mu).$

\par
As $R$ runs through the set of rank-one projections which have their support contained in $\supp{S},$ $\tr RP$ and $\tr RQ$ take all values such that $0 \leq \tr RP, \tr RQ \leq 1$ and $\tr RP+ \tr RQ=1.$ $f$ is strictly convex, hence $f'$ is strictly monotone increasing. By the strict monotonicity of $f'$
$$
\mathrm{max}_{\{R \in \cP_1(\cH): \supp{R}\subseteq \supp{S}\}} H_f(R,S)=-f'(\lambda)+C_{f, \lambda, \mu}
$$
which maximum is taken only at $R=P$ and

$$
\mathrm{min}_{\{R \in \cP_1(\cH): \supp{R}\subseteq \supp{S}\}} H_f(R,S)=-f'(\mu)+C_{f, \lambda, \mu}
$$
which minimum is taken only at $R=Q.$
Clearly,
$$
\mathrm{max}_{\{R \in \cP_1(\cH): \supp{R}\subseteq \supp{S}\}} H_f(R,S)
$$
$$
-\mathrm{min}_{\{R \in \cP_1(\cH): \supp{R}\subseteq \supp{S}\}} H_f(R,S)
$$
$$
=f'(\mu)-f'(\lambda).
$$
$f'$ is strictly monotone increasing, so $f'(\mu)-f'(\lambda)=f'(1-\lambda)-f'(\lambda)$ is a strictly monotone decreasing function of $\lambda.$ Therefore, $f'(\mu)-f'(\lambda)$ uniquely determines $\lambda$ and hence the spectrum of the rank-two density $S.$
\par
This means that the Bregman $f$-divergence preserving property of $\phi$ implies that $\phi(S)$ is a rank-two density with eigenvalues $\lambda$ and $\mu$ (and possibly zero). Hence $\phi(S)=\lambda P'+\mu Q'$ with some rank-one projections $P'$ and $Q'$ which are orthogonal to each other.
\par
Therefore, we can conclude that
$$
R=P \Leftrightarrow H_f(R,S)=-f'(\lambda)+C_{f,\lambda,\mu} \Leftrightarrow H_f\ler{\phi(R),\phi(S)}=-f'(\lambda)+C_{f,\lambda,\mu} \Leftrightarrow
$$
$$
\Leftrightarrow H_f \ler{\phi(R), \lambda P'+ \mu Q'}= -f'(\lambda)+C_{f,\lambda,\mu} \Leftrightarrow \phi(R)=P'.
$$
We deduced that $\phi(P)=P'.$ Similarly, $\phi(Q)=Q'.$ So 
\be \label{qlin}
\phi(S)=\phi\ler{\lambda P+ \mu Q}=\lambda \phi(P) + \mu \phi(Q),
\ee
and the mutual orthogonality of rank-one projections is preserved.

\par

Let $P$ and $Q$ be arbitrary mutually orthogonal elements of $\cP_1(\cH)$ and set $R \in \cP_1(\cH)$ such that $\supp{R} \subseteq \supp{P}+\supp{Q}.$ Let $0<\lambda<\mu<1$ with $\lambda + \mu=1.$ Then by (\ref{szamol}), (\ref{qlin}) and by the preserver property of $\phi$
$$
H_f(R, \lambda P + \mu Q)
$$
$$
=\ler{-f'(\lambda)+C_{f, \lambda, \mu}}\tr R P +\ler{-f'(\mu)+C_{f, \lambda, \mu}}\tr R Q
$$
$$
=H_f\ler{\phi(R), \lambda \phi(P) + \mu \phi(Q)}
$$
$$
=\ler{-f'(\lambda)+C_{f, \lambda, \mu}}\tr \phi(R) \phi(P) +\ler{-f'(\mu)+C_{f, \lambda, \mu}}\tr \phi(R) \phi(Q)
$$
We used that $\phi(R) \in \cP_1(\cH)$ such that
$$\supp{\phi(R)}\subseteq \supp{\lambda \phi(P) + \mu \phi(Q)}=\supp{\phi(P)}+\supp{\phi(Q)}.$$
Any element of a nontrivial compact real interval is a unique convex combination of the endpoints, hence we get that
\be \label{trpr}
\tr R P = \tr \phi(R) \phi(P) \text{ and } \tr R Q = \tr \phi(R) \phi(Q).
\ee
\emph{Wigner's theorem} states that any bijection $\xi: \cP_1(\cH) \rightarrow \cP_1(\cH)$ which preserves the \emph{transition probability} --- i. e., for which $\tr \xi(P)\xi(Q)=\tr P Q$ holds for any $P, Q \in \cP_1(\cH)$ --- is implemented by a unitary or antiunitary operator --- see e. g. \cite{ml99}. So, by Wigner's theorem, we get that
\be \label{p1en}
\phi(R)=U R U^* \qquad  \ler{R \in \cP_1(\cH)}
\ee
for some unitary or antiunitary operator $U$ acting on $\cH.$
\par
Now let
$$
\psi(D)=U^*\phi(D)U \qquad \ler{D \in \cS(\cH)}.
$$
Then $\psi$ is the identity on $\cP_1(\cH)$ and it preserves the Bregman $f$-divergence.
Note that by (\ref{tart1}) for any $D \in \cS(\cH)$ and $P \in \cP_1(\cH)$
$$
\supp{P} \subseteq \supp{D} \Leftrightarrow \supp{\psi(P)} \subseteq \supp{\psi(D)}
$$
and
$$
P=\psi(P) \Rightarrow \supp{P}=\supp{\psi(P)},
$$
hence 
$$\supp{D}=\supp{\psi(D)}.$$
By (\ref{compvar}), for any $D \in \cS(\cH)$ and $P, Q \in \cP_1(\cH)$ with $\supp{P} \subseteq \supp{D}$, $\supp{Q} \subseteq \supp{D}$
we have
$$
H_f(P, D)=\tr_{\supp{D}} \ler{f(P)-f(D)-f'(D)(P-D)}
$$
$$
H_f(Q, D)=\tr_{\supp{D}} \ler{f(Q)-f(D)-f'(D)(Q-D)}
$$
Using that $f(0)=f(1)=0$ and hence $f(P)=f(Q)=0,$ we get
\be \label{kul1}
H_f(P, D)-H_f(Q, D)= \tr_{\supp{D}} f'(D)(Q-P).
\ee
Similarly,
\be \label{kul2}
H_f(P, \psi(D))-H_f(Q, \psi(D))= \tr_{\supp{D}} f'\ler{\psi(D)}(Q-P).
\ee
$\psi$ preserves the Bregman $f$-divergence, hence  subtracting (\ref{kul2}) from (\ref{kul1}) one gets that that
\be \label{konz}
\tr_{\supp{D}}\ler{f'(D)-f'\ler{\psi(D)}}(Q-P)=0
\ee
for any $P,Q \in \cP_1(\cH)$ with $\supp{P} \subseteq \supp{D}, \supp{Q} \subseteq \supp{D}.$
\par
From now, unless stated otherwise, we restrict ourselves to $\supp{D}.$ 
It follows from (\ref{konz}) that
\be \label{ci}
f'(D)-f'\ler{\psi(D)}=cI
\ee
for some $c \in \R$ (in particular, $f'(D)$ and $f'\ler{\psi(D)}$ commute). Indeed, if $P$ and $Q$ are projections corresponding to different eigenvectors of $f'(D)-f'\ler{\psi(D)}$ (which is self-adjoint), then (\ref{konz}) shows that the eigenvalues are the same. So all the eigenvalues of $f'(D)-f'\ler{\psi(D)}$ are equal.
\par
Suppose that $c\neq 0,$ for example, $c>0.$ Then by (\ref{ci}),
$$
f'(D) > f'\ler{\psi(D)}.
$$
$f'(D)$ and $f'\ler{\psi(D)}$ commute and $f'$ is monotone, hence we get that
$$
D>\psi(D).
$$
This is a contradiction, so $c=0,$ that is, $f'(D)=f'\ler{\psi(D)}.$ By the strict monotonicity of $f'$ this implies $D=\psi(D).$
\par
So we deduced that $D=\psi(D)$ on $\supp{D}=\supp{\psi(D)}$ which means that without any restriction, we have
$$
D=\psi(D).
$$
We deduced that $\psi$ is the identity of $\cS(\cH),$ the proof is done.
\end{proof}

\begin{proof}[Case II]
Now we investigate the case when $\lim_{x \to 0} f'(x) \in \R.$
In this case the Bregman divergence of any two states is finite, hence it is reasonable to define
\be \label{Mdef}
M(X):=\mathrm{max}\left\{H_f(X,D) \, | D \in \cS(\cH)\right\}.
\ee
Let us note that $f$ is contiuously differentiable, hence the map $D \mapsto H_f(X,D)$ is contiuous on the compact set $\cS(\cH).$ Therefore, the above definition is correct as $\mathrm{max}\left\{H_f(X,D) \, | D \in \cS(\cH)\right\}$ exists.
In the followings we show that for $P \in \cS(\cH)$
$$
M(P)= \mathrm{max}\left\{M(X) | X \in \cS(\cH)\right\}
$$
if and only if $P$ is a pure state, i. e., a rank-one projection.
\par
Indeed, assume that $\rho$ is not a pure state, that is, $\rho=\sum_{j=1}^{k} \lambda_j P_j$ for some $k \geq 2,$ for some real numbers $0 <\lambda_1, \dots, \lambda_k<1$ with $\sum_{j=1}^{k} \lambda_j=1$ and for some rank-one projections $P_1, \dots, P_k.$ Assume that $D^* \in \cS(\cH)$ has the property that
$$
M(\rho)=\mathrm{max}\left\{H_f(\rho,D) \, | D \in \cS(\cH)\right\}=H_f(\rho, D^*).
$$
By the strict convexity of $f,$ the map
$$
X \mapsto H_f(X,Y)= \tr \ler{f(X)-f(Y)-f'(Y)(X-Y)}
$$
is strictly convex on $\cS(\cH),$ see e. g. \cite[2.10. Theorem]{carlen}. So
$$
H_f(\rho, D^*)=H_f\ler{\sum_{j=1}^{k} \lambda_j P_j, D^*}<\sum_{j=1}^{k} \lambda_j H_f(P_j, D^*).
$$
Therefore, $H_f\ler{\rho, D^*}<H_f\ler{P_{j^*},D^*}$ for some $j^* \in \{1, \dots, k\}.$
This means that
$$
M(\rho)=H_f\ler{\rho, D^*} < H_f\ler{P_{j^*},D^*}\leq M\ler{P_{j^*}}.
$$
On the other hand, it can be easily seen --- for example, by the unitary invariance of the Bregman divergences --- that $M(P)=M(Q)$ for all $P, Q \in \cP_1(\cH).$ Therefore, $M(P)$ is maximal, if $P \in \cP_1(\cH).$
\par
So we have the following characterization of the pure states.
\be \label{pschar}
P \in \cP_1(\cH) \Leftrightarrow M(P)= \mathrm{max}\left\{M(X) | X \in \cS(\cH)\right\}.
\ee
If $\phi$ is a bijection that preserves the Bregman $f$-divergence, then for any $X \in \cS(\cH)$
$$
M\ler{\phi(X)}=\mathrm{max}\left\{H_f\ler{\phi(X),\phi(D)} \, | D \in \cS(\cH)\right\}
$$
$$
=\mathrm{max}\left\{H_f\ler{X,D} \, | D \in \cS(\cH)\right\}=M(X).
$$
So $\phi$ restricted to $\cP_1(\cH)$ is a $\cP_1(\cH) \rightarrow \cP_1(\cH)$ bijection.
\par
Now let $P, Q \in \cP_1(\cH).$ By the computation rule (\ref{comp_2})
$$
H_f(P,Q)=\ler{f(0)-f(0)-f'(0)(0-0)}\tr (I-P)(I-Q)
$$
$$
+\ler{f(1)-f(0)-f'(0)(1-0)}\tr (P)(I-Q)
$$
$$
+\ler{f(0)-f(1)-f'(1)(0-1)}\tr (I-P)(Q)+\ler{f(1)-f(1)-f'(1)(1-1)}\tr PQ
$$
\be \label{egysz}
=\ler{1-\tr PQ}\ler{f'(1)-f'(0)}.
\ee
Similarly,
\be \label{egysz2}
H_f \ler{\phi(P),\phi(Q)}=\ler{1-\tr \phi(P)\phi(Q)}\ler{f'(1)-f'(0)}.
\ee
$f'$ is strictly monotone, hence $\ler{f'(1)-f'(0)} \neq 0.$ So by (\ref{egysz}) and (\ref{egysz2}), $H_f \ler{\phi(P),\phi(Q)}=H_f(P,Q)$ gives us
\be \label{trpr2}
\tr PQ = \tr \phi(P) \phi(Q) \qquad  \ler{P, Q \in \cP_1(\cH)}.
\ee
From now on, our argument is very similar to the ending part of the discussion of \emph{Case I}.
By Wigner's theorem,
\be \label{p1en_2}
\phi(R)=U R U^* \qquad  \ler{R \in \cP_1(\cH)}
\ee
for some unitary or antiunitary operator $U.$
\par
Now let
$$
\psi(D)=U^*\phi(D)U \qquad \ler{D \in \cS(\cH)}.
$$
Then $\psi$ is the identity on $\cP_1(\cH)$ and it preserves the Bregman $f$-divergence.
By (\ref{comp_2var}), for any $D \in \cS(\cH)$ and $P, Q \in \cP_1(\cH)$
$$
H_f(P, D)=\tr \ler{f(P)-f(D)-f'(D)(P-D)}
$$
$$
H_f(Q, D)=\tr \ler{f(Q)-f(D)-f'(D)(Q-D)}
$$
It follows that
\be \label{kul1_2}
H_f(P, D)-H_f(Q, D)= \tr f'(D)(Q-P).
\ee
Similarly,
\be \label{kul2_2}
H_f(P, \psi(D))-H_f(Q, \psi(D))= \tr f'\ler{\psi(D)}(Q-P).
\ee
$\psi$ preserves the Bregman $f$-divergence, hence (\ref{kul1_2}) and (\ref{kul2_2}) imply that
\be \label{konz_2}
\tr\ler{f'(D)-f'\ler{\psi(D)}}(Q-P)=0
\ee
for any $P,Q \in \cP_1(\cH).$
\par

It follows from (\ref{konz_2}) that
\be \label{ci_2}
f'(D)-f'\ler{\psi(D)}=cI
\ee
for some $c \in \R,$ and it is easy to show that $c=0$ by necessity.
\par
That is, $f'(D)=f'\ler{\psi(D)}.$ By the strict monotonicity of $f'$ this implies
$$D=\psi(D).$$
\par
We deduced that $\psi$ is the identity of $\cS(\cH),$ the proof is done.
\end{proof}

\subsection{The proof of Theorem \ref{fo_jen}}
Our aim is to prove that any bijective transformation of $\cS(\cH)$ which preserves the Jensen $f$-divergence is implemented by a unitary or an antiunitary operator. Recall that $J_f(.,.)$ denotes the Jensen $f$-divergence, which quantity was defined in Subsection \ref{jendef}.
\begin{lemma}
For any rank-one projections $P, Q \in \cP_1(\cH)$ we have
\be \label{jepro}
J_f(P,Q)=(-1)\ler{f\ler{\frac{1}{2}\ler{1+\sqrt{\tr PQ}}}+f\ler{\frac{1}{2}\ler{1-\sqrt{\tr PQ}}}}.
\ee
\end{lemma}

\begin{proof}
Let $\{e_1, \dots, e_n\} \subseteq \cH$ be an orthonormal basis such that $\supp{P}=\spn{e_1}$ and $\supp{Q}\subseteq \spn{e_1,e_2}.$
In this basis
\[
 P=
 \left[
 \ba{ccccc}
 1 & 0 & 0 & \hdots & 0 \\
 0 & 0 & 0 & \hdots & 0 \\
 0 & 0 & 0 & \hdots & 0 \\
 \vdots & \vdots & \vdots & \ddots & \vdots \\
  0 & 0 & 0 & \hdots & 0
  \ea
 \right]
\]
and
\[
 Q=
 \left[
 \ba{ccccc}
 p & \sqrt{p(1-p)}\overline{\alpha} & 0 & \hdots & 0 \\
 \sqrt{p(1-p)}\alpha & 1-p & 0 & \hdots & 0 \\
 0 & 0 & 0 & \hdots & 0 \\
 \vdots & \vdots & \vdots & \ddots & \vdots \\
  0 & 0 & 0 & \hdots & 0
  \ea
 \right] 
 \]
where we introduced the notation $p:=\tr PQ$ and $\alpha \in \C, \abs{\alpha}=1.$ Let us denote by $\mu$ and $1-\mu$ the eigenvalues of $\frac{1}{2}(P+Q)$ restricted to $\spn{e_1,e_2}.$ Easy computations show that
$$
\det\ler{\frac{1}{2}(P+Q)_{|\spn{e_1,e_2}}}=\frac{1}{4}(1-p).
$$
Therefore, $\mu(1-\mu)=\frac{1}{4}(1-p).$ The solution of this quadratic equation is
$$
\mu_{1,2}=\frac{1}{2}\ler{1 \pm \sqrt{p}},
$$
which gives the result of the lemma. (We used that by $f(0)=f(1)=0$ we have $\tr f(P)=\tr f(Q)=0.$)
\end{proof}

\begin{corollary} \label{pmax}
By the convexity of $f,$
$$
\max_{P,Q \in \cP_1(\cH)} J_f(P,Q)=-2 f \ler{\frac{1}{2}}=:M_f
$$
and by the strict convexity of $f,$ for $P,Q \in \cP_1(\cH)$ we have $J_f(P,Q)=M_f$ if and only if $\tr PQ=0,$ that is, $PQ=0.$
\end{corollary}

\begin{claim} \label{smax}
$$
\max_{A,B \in \cS(\cH)} J_f(A,B)=M_f
$$
and for any $A,B \in \cS(\cH),$ if $J_f(A,B)=M_f,$ then $AB=0.$
\end{claim}

\begin{proof}
The function $f$ is an element of the Matrix Entropy Class, hence by \cite[Thm. 2]{pv15} the induced Bregman $f$-divergence (denoted by $H_f$) is jointly convex.
Observe that
$$
J_f(A,B)=\frac{1}{2}\ler{H_f\ler{A, \frac{A+B}{2}}+H_f\ler{B, \frac{A+B}{2}}}.
$$
Suppose that
$$
A=\sum_i \lambda_i P_i, \text{ and } B= \sum_j \mu_j Q_j
$$
where the $P_i$'s and the $Q_j$'s are rank-one projections and the $\lambda_i$'s and the $\mu_j$'s are positive numbers such that $\sum_i \lambda_i=\sum_j \mu_j=1.$
The joint convexity of the Bregman divergence implies that
$$
J_f(A,B)=\frac{1}{2}\ler{H_f\ler{A, \frac{A+B}{2}}+H_f\ler{B, \frac{A+B}{2}}}
$$
$$
=\frac{1}{2}\ler{H_f\ler{\sum_i \lambda_i P_i, \sum_i \lambda_i \frac{P_i+B}{2}}+H_f\ler{\sum_i \lambda_i B, \sum_i \lambda_i \frac{P_i+B}{2}}}
$$
$$
\leq
\sum_i \lambda_i \ler{ \frac{1}{2} H_f\ler{P_i, \frac{P_i+B}{2}}+ \frac{1}{2} H_f\ler{ B, \frac{P_i+B}{2}}}=\sum_i \lambda_i J_f(P_i, B).
$$
Similarly, $J_f(P_i,B) \leq \sum_j \mu_j J_f\ler{P_i, Q_j}$ for any $i.$
Therefore,
$$
J_f(A,B) \leq \sum_i \sum_j \lambda_i \mu_j J_f(P_i, Q_j) \leq \max_{P,Q \in \cP_1(\cH)} J_f(P,Q)=M_f.
$$
We have $\lambda_i \mu_j>0$ for any $i$ and $j,$ and $\sum_i \sum_j \lambda_i \mu_j=1,$ hence if $J_f(A,B)=M_f,$ then $J_f(P_i, Q_j)=M_f$ for any $i$ and $j.$ By Corollary \ref{pmax} this means that we always have $P_iQ_j=0,$ and therefore $AB=0.$
\end{proof}
\begin{claim}
If $\phi$ is a bijection on $\cS(\cH)$ that preserves the Jensen $f$-divergence, then $\phi$ restricted to $\cP_1(\cH)$ is a bijection from $\cP_1(\cH)$ to $\cP_1(\cH).$
\end{claim}

\begin{proof}
If $P \in \cP_1(\cH),$ then there exists a set $H \subseteq \cS(\cH)$ such that $P \in H, \, \abs{H}=\dim(\cH)$ and $J_f(A,B)=M_f$ for any $A,B \in H.$ ($H$ contains orthogonal rank-one projections.) $\phi$ preserves the Jensen $f$-divergence, hence $J_f(X,Y)=M_f$ for any $X, Y \in \phi(H),$ that is, by Lemma \ref{smax}, $\phi(H) \subseteq \cS(\cH)$ has $\dim(\cH)$ pairwise orthogonal elements. This implies that all the elements of $\phi(H)$ are rank-one projections, in particular, $\phi(P) \in \cP_1(\cH).$
\end{proof}

\begin{claim}
There exists a unitary or antiunitary operator $U: \cH \rightarrow \cH$ such that
$$
\phi(R)=URU^* \qquad  \ler{R \in \cP_1(\cH)}.
$$
\end{claim}

\begin{proof}
By the strict convexity of $f,$ (\ref{jepro}) shows that for $P, Q \in \cP_1(\cH),$ $J_f(P,Q)$ is a strictly monotone decreasing function of $\tr PQ.$ Therefore, if $\phi$ preserves the Jensen $f$-divergence, then it preserves the transition probability $\tr PQ,$ as well. So $\phi$ restricted to $\cP_1(\cH)$ is a bijection from $\cP_1(\cH)$ to $\cP_1(\cH)$ which preserves the transition probability ($\tr \phi(P)\phi(Q)=\tr PQ$ for any $P, Q \in \cP_1(\cH)$). Therefore, by Wigner's theorem we obtain the statement of this Claim.
\end{proof}

\begin{claim}
Let us define the map $\psi: \cS(\cH) \rightarrow \cS(\cH)$ by
$$
\psi(D):=U^*\phi(D)U.
$$
Then $\psi(D)=D$ for any $D \in \cS(\cH).$
\end{claim}

\begin{proof}
Clearly, $\psi$ preserves the Jensen $f$-divergence and $\psi(R)=R$ for any $R \in \cP_1(\cH).$ Therefore, we have
$$
J_f(R, D)=J_f\ler{R, \psi(D)} \qquad R \in \cP_1(\cH), \, D \in \cS(\cH), 
$$
which equation can be written as
\be \label{jocon}
\tr f \ler{\frac{1}{2}\ler{R+\psi(D)}}-\tr f \ler{\frac{1}{2}\ler{R+D}}
=\frac{1}{2} \tr f \ler{\psi(D)}-\frac{1}{2} \tr f (D).
\ee
Let $D \in \cS(\cH)$ be arbitrary but fixed. Suppose that the spectral decompositions of $D$ and $\psi(D)$ are
\be \label{spec}
D=\sum_{i=1}^{m} \lambda_i P_i \text{ and } \psi(D)=\sum_{j=1}^{n} \mu_j Q_j,
\ee
where $\lambda_1>\lambda_2>\dots>\lambda_m$ and $\mu_1> \mu_2> \dots > \mu_n.$
\par
$f$ is strictly convex, hence the difference quotient function $h(a,b):=\frac{f(a)-f(b)}{a-b}$ is strictly monotone increasing in both $a$ and $b.$

Now, we show that for any $R \in \cP_1(\cH)$ the quantity $\tr f \ler{\frac{1}{2}\ler{R+D}}$ is maximal, that is,
$$
\tr f \ler{\frac{1}{2}\ler{R+D}}=\max_{X \in \cP_1(\cH)} \tr f \ler{\frac{1}{2}\ler{X+D}}
$$
if and only if $R \leq P_1,$ and in this case
$$
\tr f \ler{\frac{1}{2}\ler{R+D}}=
f\ler{\frac{\lambda_1}{2}+\frac{1}{2}}-f\ler{\frac{\lambda_1}{2}}
+\sum_{i=1}^{m} \rnk{P_i} f\ler{\frac{\lambda_i}{2}}.
$$
Indeed, let us denote the dimension of the Hilbert space $\cH$ by $N$ and let $\nu_1 \geq \nu_2 \geq \dots \geq \nu_N$ denote the (not necessarily different) eigenvalues of $D.$ As $R\geq 0,$ by Weyl's inequality (see e. g. \cite[Thm. III.2.1]{bhat}) we get that the eigenvalues of $R+D$ can be written in the form $\nu_1+\eps_1,\nu_2+\eps_2, \dots, \nu_N+\eps_N,$ where $\eps_k\geq 0$ for any $k \in \{1,\dots,N\}.$ The condition $\tr R=1$ ensures that $\sum_{k=1}^N \eps_k=1.$ Obviously,
$$
\tr f \ler{\frac{1}{2}\ler{R+D}}-\tr f \ler{\frac{1}{2}D}
=\sum_{k=1}^N f\ler{\frac{\nu_k+\eps_k}{2}}-f\ler{\frac{\nu_k}{2}}
$$
$$
=\sum_{k: \, \eps_k> 0} \frac{\eps_k}{2}
\frac{f\ler{\frac{\nu_k+\eps_k}{2}}-f\ler{\frac{\nu_k}{2}}}{\frac{\eps_k}{2}}.
$$
If $R \nleq P_1,$ then $\eps_l, \eps_k>0$ for some $1 \leq l \neq k \leq N.$ In this case, by the strict monotonicity of the difference quotient function $h$ we have
$$
\sum_{k: \, \eps_k>0} \frac{\eps_k}{2}
\frac{f\ler{\frac{\nu_k+\eps_k}{2}}-f\ler{\frac{\nu_k}{2}}}{\frac{\eps_k}{2}}
<
\sum_{k: \, \eps_k>0} \frac{\eps_k}{2}
\frac{f\ler{\frac{\ler{\nu_1+\sum_{l=1}^{k-1}\eps_l}+\eps_k}{2}}-f\ler{\frac{\nu_1+\sum_{l=1}^{k-1}\eps_l}{2}}}{\frac{\eps_k}{2}}
$$
$$
=\sum_{k=1}^N
f\ler{\frac{\ler{\nu_1+\sum_{l=1}^{k-1}\eps_l}+\eps_k}{2}}-f\ler{\frac{\nu_1+\sum_{l=1}^{k-1}\eps_l}{2}}=f\ler{\frac{\nu_1+1}{2}}-f\ler{\frac{\nu_1}{2}}.
$$
If $R \leq P_1,$ then the eigenvalues of $R+D$ are $\nu_1+1,\nu_2, \dots, \nu_N.$ Therefore, in this case we have
$$
\tr f \ler{\frac{1}{2}\ler{R+D}}-\tr f \ler{\frac{1}{2}D}=
f\ler{\frac{\nu_1+1}{2}}-f\ler{\frac{\nu_1}{2}}.
$$
So we get that $\tr f \ler{\frac{1}{2}\ler{R+D}}$ is maximal if and only if $R \leq P_1.$
\par

Similarly,
$\tr f \ler{\frac{1}{2}\ler{R+\psi(D)}}$ is maximal, that is,
$$
\tr f \ler{\frac{1}{2}\ler{R+\psi(D)}}=\max_{X \in \cP_1(\cH)} \tr f \ler{\frac{1}{2}\ler{X+\psi(D)}}
$$
if and only if $R \leq Q_1,$ and in this case
$$
\tr f \ler{\frac{1}{2}\ler{R+\psi(D)}}=
f\ler{\frac{\mu_1}{2}+\frac{1}{2}}-f\ler{\frac{\mu_1}{2}}
+\sum_{j=1}^{n} \rnk{Q_j} f\ler{\frac{\mu_j}{2}}.
$$
Observe that the right hand side of (\ref{jocon}) is independent of $R,$ hence $\tr f \ler{\frac{1}{2}\ler{R+D}}$ is maximal if and only if $\tr f \ler{\frac{1}{2}\ler{R+\psi(D)}}$ is maximal. That is, $R \leq P_1 \Leftrightarrow R \leq Q_1,$ so $P_1=Q_1.$
\par
Let us introduce the notation $S_1:=P_1=Q_1.$
For any $R \in \cP_1(\cH)$ with $R S_1=0$ $\tr f \ler{\frac{1}{2}\ler{R+D}}$ is maximal, that is,
$$
\tr f \ler{\frac{1}{2}\ler{R+D}}=\max_{X \in \cP_1(\cH): X S_1=0} \tr f \ler{\frac{1}{2}\ler{X+D}}
$$
if and only if $R \leq P_2,$ and in this case
$$
\tr f \ler{\frac{1}{2}\ler{R+D}}=
f\ler{\frac{\lambda_2}{2}+\frac{1}{2}}-f\ler{\frac{\lambda_2}{2}}
+\sum_{i=1}^{m} \rnk{P_i} f\ler{\frac{\lambda_i}{2}}.
$$
Similarly,
$\tr f \ler{\frac{1}{2}\ler{R+\psi(D)}}$ is maximal, that is,
$$
\tr f \ler{\frac{1}{2}\ler{R+\psi(D)}}=\max_{X \in \cP_1(\cH): X S_1=0} \tr f \ler{\frac{1}{2}\ler{X+\psi(D)}}
$$
if and only if $R \leq Q_2,$ and in this case
$$
\tr f \ler{\frac{1}{2}\ler{R+\psi(D)}}=
f\ler{\frac{\mu_2}{2}+\frac{1}{2}}-f\ler{\frac{\mu_2}{2}}
+\sum_{j=1}^{n} \rnk{Q_j} f\ler{\frac{\mu_j}{2}}.
$$
The right hand side of (\ref{jocon}) is independent of $R,$ hence $\tr f \ler{\frac{1}{2}\ler{R+D}}$ is maximal if and only if $\tr f \ler{\frac{1}{2}\ler{R+\psi(D)}}$ is maximal. That is, $R \leq P_2 \Leftrightarrow R \leq Q_2,$ so $P_2=Q_2.$
\par
And so on, we can deduce that all the eigenprojections coincide, that is, $m=n$ and $P_i=Q_i$ for all $1\leq i\leq m.$
\par
Let us define the following function on the positive half line.
$$
g(a):=f\ler{\frac{a}{2}+\frac{1}{2}}-f\ler{\frac{a}{2}}
$$
By the strict convexity of $f,$ $g$ is strictly monotone increasing. It follows easily from the above variational formulas that
$$
\max_{R \in \cP_1(\cH): R \leq P_{l}+P_{l+1}} \tr f \ler{\frac{1}{2}\ler{R+D}}-
\min_{R \in \cP_1(\cH): R \leq P_{l}+P_{l+1}} \tr f \ler{\frac{1}{2}\ler{R+D}}
$$
$$
=g\ler{\lambda_l}-g\ler{\lambda_{l+1}}
$$
and
$$
\max_{R \in \cP_1(\cH): R \leq P_{l}+P_{l+1}} \tr f \ler{\frac{1}{2}\ler{R+\psi(D)}}-
$$
$$
-\min_{R \in \cP_1(\cH): R \leq P_{l}+P_{l+1}} \tr f \ler{\frac{1}{2}\ler{R+\psi(D)}}
$$
$$
=g\ler{\mu_l}-g\ler{\mu_{l+1}}
$$
for all $1\leq l  \leq m-1.$
The right hand side of (\ref{jocon}) is independent of $R,$ hence we have
$$
g\ler{\lambda_l}-g\ler{\lambda_{l+1}}=g\ler{\mu_l}-g\ler{\mu_{l+1}}
$$
for all $1\leq l  \leq m-1.$
\par
Indirectly assume that $\lambda_k>\mu_k$ for some $1 \leq k\leq m.$ Then by $\sum_{i=1}^m \lambda_i=\sum_{i=1}^m \mu_i=1$ we have $\lambda_s<\mu_s$ for some $1 \leq s\leq m.$ Without loss of generality we may assume that $s>k.$ Then there exists some $t \in \{k, k+1, \dots, s-1\}$ such that
$$
\left[\lambda_{t+1},\lambda_t\right] \supsetneq \left[\mu_{t+1},\mu_t\right]
$$
which implies by the strict monotonicity of $g$ that
$$
g\ler{\lambda_t}-g\ler{\lambda_{t+1}}>g\ler{\mu_t}-g\ler{\mu_{t+1}}.
$$
A contradiction.
\par
So the proof of the claim is done, and hence the proof Theorem \ref{fo_jen} is complete. 
\end{proof}

\subsection*{Acknowledgement}
The author is grateful to Lajos Moln\'ar for illuminating discussions.

\bibliographystyle{amsplain}

\end{document}